\documentclass[
submission
]{dmtcs-episciences}

\usepackage[utf8]{inputenc}
\usepackage{subfigure}

\newtheorem{theorem}{Theorem}
\newtheorem{definition}{Definition}

\usepackage[round]{natbib}

\author{Mircea Parpalea\affiliationmark{1}
  \and Nicoleta Avesalon\affiliationmark{2}
  \and Eleonor Ciurea\affiliationmark{2}}
\title[Minimum parametric flow over time] {Minimum parametric flow over time}
\affiliation{
  Andrei \c Saguna National College, Bra\c sov, Romania\\
  Transilvania University of Bra\c sov, Romania}
\begin{document}
\publicationdetails{VOL}{2015}{ISS}{NUM}{SUBM}
\maketitle
\begin{abstract}
  The paper presents a dynamic solution method for dynamic minimum parametric networks flow. The solution method solves the problem for a special parametric dynamic network with linear lower bound functions of a single parameter. Instead directly work on the original network, the method implements a labelling algorithm in the parametric dynamic residual network and uses quickest paths from the source node to the sink node in the time-space network along which repeatedly decreases the dynamic flow for a sequence of parameter values, in their increasing order. In each iteration, the algorithm computes both the minimum flow for a certain subinterval of the parameter values, and the new breakpoint for the maximum parametric dynamic flow value function.
\end{abstract}

\section{Introduction}
\label{sec:in}
Dynamic flow problems where networks structure changes depending on a scalar parameter $\lambda$ are widely used (see for example \cite{2}) to model different network-structured, decision-making problems over time. These types of problems are arising in various real applications such as communication networks, air/road traffic control, and production systems. Moreover, in many applications of graph algorithms, including communication networks, graphics, assembly planning, and scheduling, graphs are subject to discrete changes, such as additions or deletions of arcs or nodes. In the last decade there has been a growing interest for such dynamically changing graphs, and a whole body of algorithms and data structures for dynamic graphs has been discovered: \cite{5}, \cite{7}, \cite{10}, \cite{14} or \cite{16}. 
Further on, the next section presents some basic discrete-time dynamic networks terminology and notations and Section \ref{sec:two} introduces the parametric minimum flow over time problem. In Section \ref{sec:three} the algorithm for solving the parametric minimum flow in discrete dynamic networks is presented and Section \ref{sec:four} ilustrates how the algorithmon work on a given dynamic network.

\clearpage

\section{Discrete-time dynamic network}
\label{sec:one}
A discrete dynamic network $G=(N,A,T)$ is a directed graph with $N$ being a set of $\vert N \vert=n$ nodes $i$, $A$ being a set of $\vert A \vert=m$ arcs $a$ and $T$ being the finite time horizon discretized into the set $H=\{0,1,\cdots, T\}$. An arc $a\in A$ from node $i$ to node $j$ is denoted by $(i,j)$. Parallel, as well as opposite arcs are not allowed in graph $G$. The following time-dependent functions are associated with each arc $a=(i,j)\in A$: the \textit{upper bound (capacity)} function $u(i,j;\theta)$, $u:A\times H \rightarrow\Re^{+}$, representing the maximum amount of flow that can enter the arc $(i,j)$ at time $\theta$, the \textit{lower bound} function $\ell(i,j;\theta)$, $\ell:A\times H \rightarrow\Re^{+}$, i.e. the minimum amount of flow that must enter the arc $(i,j)$ at time $\theta$, and the \textit{transit time} function $h(i,j;\theta)$, $h:A\times H \rightarrow \aleph$. Time is measured in discrete steps, so that if one unit of flow leaves node $i$ at time $\theta$ over the arc $a=(i,j)$ that one unit of flow arrives at node $j$ at time $\theta+h(i,j;\theta)$, where $h(i,j;\theta)$ is the transit time of the arc $(i,j)$. The time horizon $T$ represents the time limit until which the flow can travel in the network. The network has two special nodes: a source node $s$ and a sink node $t$.

\subsection{Time-space network}
\label{sec:one_1} 
For a given discrete-time dynamic network, the \textit{time-space network} is a static network constructed by expanding the original network in the time dimension, by considering a separate copy of every node $i\in N$ at every discrete time step $\theta \in H$. A \textit{node-time pair} (NTP) $(i, \theta)$ refers to a particular node $i \in N$ at a particular time step $\theta \in H$, i.e., $(i,\theta)\in N \times H$.

\begin{definition} \label{D1}
The \textit{time-space network} $G^{T}$ of the original dynamic network $G$ is defined as follows:\ \newline
$(a) \;  N^{T}:=\{ (i,\theta) \vert i \in N, \theta \in H \}$; \newline
$(b) \;  A^{T}:=\{ a_{\theta}=((i,\theta),(j,\theta+h(i,j;\theta))) \vert 0 \leq \theta \leq T-h(i,j;\theta), (i,j)\in A \}$; \newline
$(c) \;  u^{T}(a_{\theta}):=u(a;\theta) \qquad for \quad a_{\theta}\in A^{T}$;  \newline
$(d) \;  \ell^{T}(a_{\theta}):=\ell(a;\theta)  \qquad for \quad  a_{\theta}\in A^{T}$.
\end{definition}

For every arc $(i,j) \in A$ with traversal time $h(i,j; \theta)$, capacity $u(i,j; \theta)$ and lower bound $\ell (i,j; \theta)$, the time-space network $G^{T}$ contains the arcs $((i,\theta),(j,\theta+h(i,j;\theta)))$, $\theta =0,1,\cdots, T-h(i,j;\theta)$ with capacities $u(i,j;\theta)$ and lower bounds $\ell(i,j;\theta)$. A (discrete-time) \textit{dynamic path} $\bar{P}$ is defined as a sequence of distinct, consecutively linked NTPs.

\subsection{Parametric dynamic network}
\label{sec:one_2} 
\begin{definition} \label{D2}
A discrete-time dynamic network $G=(N,A,T)$ for which the lower bounds $\ell(i,j;\theta)$ of some arcs $(i,j)\in A$ are functions of a real parameter $\lambda$ is referred to as a \textit{parametric dynamic network} and is denoted by $\bar{G}=(N,A,u,\bar{\ell},h,T)$.
\end{definition}
	
For a parametric dynamic network $\bar{G}$, the parametric lower bound function $\bar{\ell}:A\times H \times [0,\Lambda]  \rightarrow\Re^{+}$ associates to each arc $(i,j)\in A$ and for each of the parameter values $\lambda$ in an interval $[0,\Lambda]$ the real number $\bar{\ell}(i,j;\theta;\lambda)$, referred to as the \textit{parametric lower bound} of arc $(i,j)$:
\begin{equation} \label{1} \bar{\ell}(i,j;\theta;\lambda)=\ell_{0}(i,j;\theta)+\lambda \cdot \mathcal{L}(i,j;\theta), \quad \lambda \in [0,\Lambda], \quad \theta\in H,  \end{equation}
where $\mathcal{L}:A\times H \rightarrow\Re$ is a real valued function, associating to each arc $(i,j)\in A$ and for each time-value $\theta \in H$ the real number $\mathcal{L}(i,j;\theta)$, referred to as the \textit{parametric part of the lower bound} of the arc $(i,j)$. The nonnegative value $\ell_{0}(i,j;\theta)$ is the lower bound of the arc $(i,j)$ for $\lambda=0$, i.e., $\bar{\ell}(i,j;\theta;0)=\ell_{0}(i,j;\theta)$. For the problem to be correctly formulated, the lower bound function of every arc $(i,j)\in A$  must respect the condition $u(i,j;\theta)\geq \bar{\ell}(i,j;\theta;\lambda)$ for the entire interval of the parameter values, i.e., $\forall (i,j)\in A$, $\forall \theta\in H$ and $\forall \lambda \in [0,\Lambda]$. It follows that $u(i,j;\theta)\geq \ell_{0}(i,j;\theta)$ and the parametric part of the lower bounds $\mathcal{L}(i,j;\theta)$ must satisfy the constraints: $\mathcal{L}(i,j;\theta)\leq [u(i,j;\theta)-\ell_{0}(i,j;\theta)]/\Lambda$, $\forall (i,j)\in A$. 

\section{Parametric flow over time}
\label{sec:two} 
The \textit{parametric dynamic flow value function} $\bar{v}: N\times H\times [0,\Lambda]\rightarrow \Re$ associates to each of the nodes $i \in N$, at each time moment $\theta \in H$, a real number $\bar{v}(i;\theta;\lambda)$ referred to as the value of node $i$ at time $\theta$, for each of the parameter $\lambda$ values.

\begin{definition} \label{D3}
A \textit{feasible parametric flow over time} $\bar{f}(i,j;\theta;\lambda)$ in a parametric dynamic network $\bar{G}=(N, A, u, \bar{\ell}, h, T)$ is a function $\bar{f}: A\times H\times [0,\Lambda]\rightarrow \Re^{+}$ that satisfies the following constraints for all $\lambda \in [0,\Lambda]$:

\begin{equation} \label{2} \sum _{j|(i,j)\in A} \bar{f}(i,j;\theta;\lambda)-\sum _{j|(j,i)\in A}\sum ^{\vartheta|\vartheta+h(j,i;\vartheta)=\theta} \bar{f}(j,i;\vartheta;\lambda)= \{ \begin{tabular}{ll} $\bar{v}(i;\theta;\lambda)$, & $i=s,t$ \\ 0, & $i\neq s,t$  \end{tabular}, \forall i \in N, \forall \theta \in H; \end{equation}

\begin{equation} \label{3} \bar{\ell}(i,j;\theta;\lambda)\leq \bar{f}(i,j;\theta;\lambda) \leq u(i,j;\theta), \qquad \forall (i,j)\in A, \quad \forall \theta\in H; \end{equation}
\begin{equation} \label{4} \bar{f}(i,j;\theta;\lambda)=0, \qquad \forall (i,j)\in A, \quad \forall \theta\in [T-h(i,j;\theta)+1,\quad T]; \end{equation}
\begin{equation} \label{5} \sum _{\theta \in H} \bar{v}(s;\theta;\lambda)= - \sum _{\theta \in H} \bar{v}(t;\theta;\lambda)=\bar{v}(\lambda); \end{equation}
where $\bar{f}(i,j;\theta;\lambda)$ determines the rate of flow (per time unit) entering arc $(i,j)$ at time $\theta$, for the parameter value $\lambda$, $\forall \theta\in \{0,1,\cdots, T\}$ and $\forall \lambda \in [0,\Lambda]$. 
  \end{definition}

\begin{definition} \label{D4} The \textit{parametric minimum flow over time} (PmFT) problem is to compute all minimum flows over time for every possible value of $\lambda$:
\begin{equation} \label{6} minimise \quad \bar{v}(\lambda)=\sum _{{\theta \in \{0,1,\cdots, T\}}} \bar{v}(s;\theta;\lambda), \quad for \quad all \quad \lambda \in [0,\Lambda],  \end{equation} under flow constraints (\ref{2})-(\ref{5}). \end{definition}

\begin{definition} \label{D5} For the minimum flow over time problem, the \textit{parametric time-dependent residual network} with respect to a given feasible parametric flow over time $\bar{f}$ is defined as $\bar{G}(\bar{f}):=(N,A(\bar{f}),T)$, with $A(\bar{f}):=A^{+}(\bar{f}) \cup A^{-}(\bar{f})$, where 
\begin{equation} \label{7} A^{+}(\bar{f}):=\{(i,j) \vert (i,j)\in A, \exists \theta \leq T-h(i,j;\theta) \quad with \quad \bar{f}(i,j;\theta;\lambda) - \bar{\ell}(i,j;\theta;\lambda) >0 \}; \end{equation}
\begin{equation} \label{8} A^{-}(\bar{f}):=\{(i,j) \vert (j,i)\in A, \exists \theta \leq T-h(j,i;\theta) \quad with \quad u(j,i;\theta)-\bar{f}(j,i;\theta;\lambda) >0 \}.   \end{equation}
\end{definition}

The direct arcs $(i,j)\in A^{+}(\bar{f})$ in $\bar{G}(\bar{f})$ have same transit times $h(i,j;\theta)$ with those in $\bar{G}$ while the reverse arcs $(i,j)\in A^{-}(\bar{f})$ have negative transit times $h(i,j;\theta+h(j,i;\theta))=-h(j,i;\theta)$, with $(j,i)\in A$ and $0 \leq \theta+h(j,i; \theta) \leq T$. 

\begin{definition} \label{D6} For the parametric minimum flow over time (PmFT) problem, the \textit{parametric residual capacities} of arcs $(i,j)$ in the parametric time-dependent residual network $\bar{G}(\bar{f})$ are defined as follows: 
\begin{equation} \label{9} \bar{r}(i,j;\theta;\lambda)=\bar{f}(i,j;\theta;\lambda) - \bar{\ell}(i,j;\theta;\lambda), \quad (i,j)\in A, \quad 0 \leq \theta+h(i,j; \theta) \leq T; \end{equation}
\begin{equation} \label{10} \bar{r}(i,j;\theta+h(j,i; \theta);\lambda)= u(j,i;\theta) - \bar{f}(j,i;\theta;\lambda) , \quad (j,i)\in A, \quad 0 \leq \theta+h(j,i; \theta) \leq T.  \end{equation}
\end{definition}

\begin{definition} \label{D7} Given a parametric flow over time $\bar{f}(i,j;\theta;\lambda)$, the \textit{parametric residual capacity} $\bar{r}(\bar{P};\theta;\lambda)$ of a dynamic path $\bar{P}$ is the inner envelope of the parametric residual capacity functions $\bar{r}(i,j;\theta;\lambda)$ for all arcs $(i,j)$ composing the path and for all parameter $\lambda$ values: 
\begin{equation} \label{11} \bar{r}(\bar{P};\theta;\lambda)= \min \{ \bar{r}(i,j;\theta;\lambda) \vert (i,j)\in \bar{P} \}. \end{equation} \end{definition}

\begin{definition} \label{D8} The \textit{transit time} $\tau(\bar{P};\theta)$ of a dynamic path $\bar{P}$ is defined by: 
\begin{equation} \label{12}  \tau(\bar{P})= \sum_{(i,j)\in \bar{P}} h(i,j;\theta). \end{equation} \end{definition}

A dynamic path $\bar{P}$ is referred to as the \textit{quickest dynamic path} if $\tau(\bar{P}) \leq \tau(\bar{P}')$ for all dynamic paths $\bar{P}'$ in the parametric time-dependent residual network $\bar{G}(\bar{f})$.

\section{Parametric minimum dynamic flow}
\label{sec:three} 
The approaches for solving the maximum parametric flow over time problem via applying classical algorithms can be grouped in two main categories: by applying a classical parametric flow algorithm (see \cite{6} and \cite{15}) in the static time-space network (for this approach see the algorithm presented in \cite{13}) or by applying a non-parametric maximum dynamic flow algorithm (see the algorithms presented in \cite{12}) in dynamic residual networks generated by partitioning the interval of the parameter values (partitioning method presented by \cite{3}). Given the powerful versatility of dynamic algorithms, it is not surprising that these algorithms and dynamic data structures are often more difficult to design and analyse than their static counterparts (see \cite{11}). \cite{4} proved that the complexity of finding a shortest dynamic flow augmenting path, by exploring the forward and reverse arcs successively, is $O(nmT^{2})$. For algorithms which explores the two sub-networks simultaneously, \cite{8} also reported a complexity of $O(n^{2}T^{2})$. By using special node addition and selection procedures, \cite{9} succeeded to reduce significantly the number of node time pair that needs to be visited. The worst-case complexity of their algorithm is $O(nT(n+T))$. As far as we know, the problem of minimum flow over time in parametric networks has not been treated yet. 

\subsection{Parametric minimum dynamic flow (PmDF) algorithm}
\label{sec:three_1}
The idea of the algorithm is that if the parametric residual capacities for all arcs in $\bar{G}(\bar{f})$ are maintained linear functions of $\lambda$, with no break points, the problem can be solved via a slightly modified non-parametric algorithm.
Firstly, if one exists, a feasible flow must be established. The most convenient is this to be done in the nonparametric network $G'=(N,A,u,\ell',h,T)$ obtained from the initial network by replacing the parametric lower bound functions with the non-parametric ones: $\ell'(i,j;\theta)=max \{ \ell(i,j;;\theta;\lambda) \vert \lambda \in [0,\Lambda] \}$, i.e. $\ell'(i,j;\theta)=\ell_{0}(i,j;\theta)$ for $\mathcal{L}(i,j;\theta)\leq 0$ and $\ell'(i,j;\theta)=\ell_{0}(i,j;\theta)+ \Lambda \cdot \mathcal{L}(i,j;\theta)$ for $L(i,j;\theta)> 0$. For finding a feasible flow $\bar{f}(i,j;\theta)$ in $\bar{G}$ see the algorithms presented in \cite{1}, applied in the static time-space network. During the running of the algorithm, the subinterval of the parameter values is continuously narrowed so that the above restriction to remain valid. In each subinterval $[\lambda_{k},\lambda_{k+1}]$ of the parameter values, the linear parametric residual capacity of every arc $(i,j)$ can be written as $\bar{r}_{k}(i; j;\theta;\lambda)=\alpha_{k}(i,j;\theta)+\beta_{k}(i,j;\theta)\cdot (\lambda -\lambda_{k})$ while the parametric flow is $\bar{f}_{k}(i; j;\theta;\lambda)=f_{k}(i,j;\theta)+F_{k}(i,j;\theta)\cdot (\lambda -\lambda_{k})$.
As soon as the parametric time-dependent residual network $\bar{G}(\bar{f})$ contains no dynamic paths, the algorithm computes the minimum flow for the considered subinterval and then reiterates on the next subinterval of the parameter values, until $\Lambda$ value is reached.  

\footnotesize \begin{flushleft}
\hspace{0,5 cm}(01) PmDF ALGORITHM;

\hspace{0,5 cm}(02) BEGIN

\hspace{0,5 cm}(03) \quad  find a feasible flow $\bar{f}(i,j;\theta)$ in network $\bar{G}$;

\hspace{0,5 cm}(04) \quad $BP:=\{0\}$; \quad  $k:=0$; \quad  $\lambda_{k}:=0$;

\hspace{0,5 cm}(05) \quad REPEAT

\hspace{0,5 cm}(06) \quad \quad QDP($k,\lambda_{k}, BP$);

\hspace{0,5 cm}(07) \quad  \quad $k:=k+1$;

\hspace{0,5 cm}(08) \quad  UNTIL($\lambda_{k}=\Lambda$);

\hspace{0,5 cm}(09) END.  

\hspace{0,5 cm}

\hspace{1,5 cm} \small Algorithm 1: Parametric minimum dynamic flow (PmDF) algorithm.  \end{flushleft} \normalsize

\subsection{Quickest Dynamic Paths (QDP) procedure}
\label{sec:three_2}
The successive shortest (quickest) dynamic paths procedure repeatedly performs the following operations:\newline
-	Finds a shortest (quickest) dynamic path $\bar{P}$ in the time-dependent residual network $\bar{G}(\bar{f})$ based on the successor vector $\sigma:N\times H \rightarrow (i,\theta)\in N \times H$;\newline
-	Computes the parametric residual capacity $\bar{r}(\bar{P};\theta;\lambda)$ of the dynamic path $\bar{P}$;\newline
-	Computes the first (in increasing order) value of the parameter $\lambda$ up to which the parametric residual capacity of the dynamic path remains linear without break points;\newline
-	Updates the subinterval of the parameter values for which the minimum flow is determined;\newline
-	Decreases the flow along the dynamic path and updates the time-dependent residual network.\newline 
The algorithm ends when none of the source nodes, at any time moments $\theta \in \{ 0,1,\cdots,T\}$, is reachable from any of the sink nodes, i.e. there is no dynamic path from $s$ to $t$.

\footnotesize \begin{flushleft}
\hspace{0,5 cm}(01) PROCEDURE QDP($k,\lambda_{k}, BP$);

\hspace{0,5 cm}(02) BEGIN

\hspace{0,5 cm}(03) \quad  FOR all $\theta \in H$ DO

\hspace{0,5 cm}(04) \quad \hspace{0,2 cm} BEGIN

\hspace{0,5 cm}(05) \quad \hspace{0,5 cm} FOR all $i \in N$ DO $\sigma(i,\theta):=(0,0)$;

\hspace{0,5 cm}(06) \quad \hspace{0,5 cm} FOR all $(i,j) \in A$ DO 

\hspace{2,4 cm} $\alpha_{k}(i,j;\theta):=\bar{f}(i,j;\theta)-\ell_{0}(i,j;\theta)-\lambda_{k}\cdot \mathcal{L}(i,j;\theta)$; $\beta_{k}(i,j;\theta):= - \mathcal{L}(i,j;\theta)$; 

\hspace{2,4 cm} $\alpha_{k}(j,i;\theta+h(i,j;\theta)):=u(i,j;\theta)-\bar{f}(i,j;\theta)$; $\beta_{k}(j,i;\theta+h(i,j;\theta)):=0$;

\hspace{2,4 cm} $f_{k}(i,j;\theta):=\bar{f}(i,j;\theta)$; $F_{k}(i,j;\theta):=0$;

\hspace{0,5 cm}(07) \quad  \hspace{0,2 cm} END;

\hspace{0,5 cm}(08) \quad  $C:=1$; $\lambda_{k+1}=\Lambda$;

\hspace{0,5 cm}(09) \quad  \textbf{LS}($\sigma,C$);

\hspace{0,5 cm}(10) \quad  WHILE($C=1$) DO

\hspace{0,5 cm}(11) \quad \hspace{0,2 cm} BEGIN

\hspace{0,5 cm}(12) \quad \hspace{0,5 cm} build $\bar{P}$ based on $\sigma$ starting from $(s,\bar{\theta})$;

\hspace{0,5 cm}(13) \quad  \hspace{0,5 cm} $\alpha := \min \{ \alpha_{k}(i,j;\theta) \vert (i,j) \in \bar{P} \}$; $\beta:=\min \{ \beta_{k}(i,j;\theta) \vert (i,j) \in \bar{P}, \; \alpha_{k}(i,j;\theta) =\alpha \}$;

\hspace{0,5 cm}(14) \quad  \hspace{0,5 cm} $(i,\theta):=(s,\bar{\theta})$;

\hspace{0,5 cm}(15) \quad  \hspace{0,5 cm} WHILE($i \neq t$) DO

\hspace{0,5 cm}(16) \quad  \hspace{0,7 cm} BEGIN

\hspace{0,5 cm}(17) \quad  \hspace{1 cm} $(j,\vartheta):=\sigma(i,\theta)$;

\hspace{0,5 cm}(18) \quad  \hspace{1 cm} IF($\beta_{k}(i,j;\theta)<\beta$) THEN $\lambda_{k+1}:=\min \{\lambda_{k+1}, \lambda_{k}+(\alpha_{k}(i,j;\theta)-\alpha)/(\beta - \beta_{k}(i,j;\theta)) \}$;

\hspace{0,5 cm}(19) \quad  \hspace{1 cm} $\alpha_{k}(i,j;\theta):=\alpha_{k}(i,j;\theta)-\alpha$; $\beta_{k}(i,j;\theta):=\beta_{k}(i,j;\theta)-\beta$;

\hspace{2,4 cm} $\alpha_{k}(j,i;\vartheta):=\alpha_{k}(j,i;\vartheta)+\alpha$; 

\hspace{2,4 cm} $\beta_{k}(j,i;\vartheta):=\beta_{k}(j,i;\vartheta)+\beta$;

\hspace{0,5 cm}(20) \quad  \hspace{1 cm}  IF($(i,j)\in A^{+}(\bar{f})$) THEN $f_{k}(i,j;\theta):=f_{k}(i,j;\theta)-\alpha$; $F_{k}(i,j;\theta):=F_{k}(i,j;\theta)-\beta$;

\hspace{0,5 cm}(21) \quad  \hspace{1 cm}  ELSE $f_{k}(j,i;\vartheta):=f_{k}(j,i;\vartheta)+\alpha$; $F_{k}(j,i;\vartheta):=F_{k}(j,i;\vartheta)+\beta$;

\hspace{0,5 cm}(22) \quad  \hspace{1 cm}   $(i,\theta):=(j,\vartheta)$;

\hspace{0,5 cm}(23) \quad  \hspace{0,7 cm} END;

\hspace{0,5 cm}(24) \quad  \hspace{0,5 cm} FOR all $\theta \in H$ DO FOR all $i \in N$ DO $\sigma(i,\theta):=(0,0)$;

\hspace{0,5 cm}(25) \quad  \hspace{0,5 cm} \textbf{LS}($\sigma,C$);

\hspace{0,5 cm}(26) \quad  \hspace{0,2 cm} END;

\hspace{0,5 cm}(27) \quad  $BP:=BP\cup \{\lambda_{k+1}\}$;

\hspace{0,5 cm}(28) END. 

\hspace{0,5 cm}

\hspace{1,5 cm} \small Algorithm 2: Quickest Dynamic Paths (QDP) procedure.  \end{flushleft}  \normalsize 

\subsection{Labels Setting (LS) procedure}
\label{sec:three_3}
The label setting procedure uses \textit{transit time labels} $\tau(i,\theta)$ associated to all nodes at each discrete time values. \newline
At any step of the algorithm, a label is \textit{permanent} once it denotes the length of shortest augmenting path to a node-time pair, otherwise it is \textit{temporary}. The \textit{LS} procedure maintains a set $L$ of candidate nodes in increasing order of their temporary labels, which initially includes only the sink nodes $(t,\theta)$, $\theta \in \{0,1,\cdots, T\}$. For every node-time pair $(j,\vartheta)$ selected from the list, the arcs with positive residual capacity connecting $(i,\theta)$ to $(j,\vartheta)$ are explored, where $\vartheta=\theta+h(i,j;\theta)$ if the arc connecting $(i,\theta)$ to $(j,\vartheta)$ is a forward arc and $0\leq \theta=\vartheta+h(j,i;\vartheta)$ if $(i,j)$ is a reverse arc. At any iteration, the algorithm selects the node-time pair $(i,\theta)$ with the minimum temporary label, makes its transit time label permanent, checks optimality conditions and updates the labels accordingly. A transit time label $\tau(i,\theta)$ represents the length (transit time) of a shortest (quickest) dynamic path from $(i,\theta)$ if $\tau(i,\theta) \leq \tau(j,\vartheta)+h(j,i;\vartheta)$, $\forall (i,j)\in A(\bar{f})$. The process is repeated until there are no more candidate nodes in $L$. The transit time of the shortest (quickest) path $\bar{P}$ computed based on successor vector $\sigma$ is given by $\bar{\tau}$.

\footnotesize \begin{flushleft}
\hspace{0,5 cm}(01) PROCEDURE LS($\sigma, C$);

\hspace{0,5 cm}(02) BEGIN

\hspace{0,5 cm}(03) \quad  FOR all $\theta \in \{0, 1, \cdots, T\}$ DO

\hspace{0,5 cm}(04) \quad \hspace{0,2 cm} BEGIN

\hspace{0,5 cm}(05) \quad \hspace{0,5 cm} FOR all $i \in N-{t}$ DO $\tau(i,\theta):=\infty$; 

\hspace{0,5 cm}(06) \quad \hspace{0,5 cm}  $\tau(t,\theta):=0$; $L:=L \cup \{(t,\theta) \}$;

\hspace{0,5 cm}(07) \quad  \hspace{0,2 cm} END;

\hspace{0,5 cm}(08) \quad  $\bar{\tau}:=\infty$; $\bar{\theta}:=\infty$

\hspace{0,5 cm}(09) \quad  WHILE($L\neq \emptyset$) DO

\hspace{0,5 cm}(10) \quad  \hspace{0,2 cm} BEGIN

\hspace{0,5 cm}(11) \quad  \hspace{0,5 cm} select the first $(j,\vartheta)$ from $L$; $L:=L-\{ (j,\vartheta) \}$;

\hspace{0,5 cm}(12) \quad  \hspace{0,5 cm} FOR all $i \in N$ with $(i,j)\in A^{+}(\bar{f})$ DO

\hspace{0,5 cm}(13) \quad  \hspace{0,6 cm} FOR all $\theta$ with $\theta + h(i,j;\theta)=\vartheta$ DO

\hspace{0,5 cm}(14) \quad  \hspace{0,7 cm} IF($\tau(j,\vartheta)+h(i,j;\theta)<\tau(i,\theta)$) THEN

\hspace{0,5 cm}(15) \quad  \hspace{1 cm} BEGIN

\hspace{0,5 cm}(16) \quad  \hspace{1,3 cm} $\tau(i,\theta):=\tau(j,\vartheta)+h(i,j;\theta)$

\hspace{0,5 cm}(17) \quad  \hspace{1,3 cm} $\sigma(i,\theta):=(j,\vartheta)$

\hspace{0,5 cm}(18) \quad  \hspace{1,3 cm} IF($(i,\theta)\notin L$) THEN $L:=L\cup \{ (i,\theta) \}$;

\hspace{0,5 cm}(19) \quad  \hspace{1 cm}  END;

\hspace{0,5 cm}(20) \quad  \hspace{0,5 cm} FOR all $i \in N$ with $(i,j)\in A^{-}(\bar{f})$ DO 

\hspace{0,5 cm}(21) \quad  \hspace{0,6 cm} IF($\vartheta+h(j,i;\vartheta) \leq T$ and $\tau(j,\vartheta)-h(j,i;\vartheta)<\tau(i,\vartheta+h(j,i;\vartheta))$) THEN

\hspace{0,5 cm}(22) \quad  \hspace{1 cm} BEGIN

\hspace{0,5 cm}(23) \quad  \hspace{1,3 cm} $\tau(i,\vartheta+h(j,i;\vartheta)):=\tau(j,\vartheta)-h(j,i;\vartheta)$;

\hspace{0,5 cm}(24) \quad  \hspace{1,3 cm} $\sigma(i, \vartheta+h(j,i;\vartheta)):=(j,\vartheta)$;

\hspace{0,5 cm}(25) \quad  \hspace{1,3 cm} IF($(i, \vartheta+h(j,i;\vartheta))\notin L$) THEN $L:=L\cup \{ (i,\vartheta+h(j,i;\vartheta)) \}$;

\hspace{0,5 cm}(26) \quad  \hspace{1 cm} END; 

\hspace{0,5 cm}(27) \quad  \hspace{0,2 cm} END;

\hspace{0,5 cm}(28) \quad  $\bar{\tau}:=\min_{\theta \in \{ 0, 1, \cdots, T \}} \{ \tau(s,\theta)\}$; $\bar{\theta}:=\min_{\theta \in \{ 0, 1, \cdots, T \}} \{ \theta \vert \tau(s,\theta)=\bar{\tau}\}$;

\hspace{0,5 cm}(29) \quad  IF($\bar{\tau}=\infty$) THEN $C:=0$;

\hspace{0,5 cm}(30) END. 

\hspace{0,5 cm}

\hspace{1,5 cm} \small Algorithm 3: Labels Setting (LS) procedure. \end{flushleft}  \normalsize

\begin{theorem}\label{T1}[Correctness of PmDF algorithm] 
PmDF algorithm computes correctly a parametric minimum flow over time for a given time horizon $T$ and for $\lambda\in[0,\Lambda]$. 
\end{theorem}
\begin{proof} The partitioning type algorithm iterates on successive subintervals $[\lambda_{k},\lambda_{k+1}]$, starting with $\lambda_{0}=0$ and ending with $\lambda_{k_{max}+1}=\Lambda$ and consequently, the correctness of the algorithm obviously follows from the correctness of the Quickest Dynamic Paths (QDP) procedure which ends when none of the source node-time pairs is reachable from any of the sink node-time pairs, i.e. when there is no dynamic path from the source node to the sink node in the time-depending residual network. According to the classical flow decreasing path theorem presented by \cite{1}, this means that the obtained flow is a minimum dynamic flow for the given time horizon. In fact, the algorithm ends with a set of linear parametric minimum flows and with the partition $BP$ of the interval of the parameter values in their corresponding subintervals. \end{proof}

\begin{theorem}\label{T2} [Time complexity of PmDF algorithm] 
The parametric minimum dynamic flow (PmDF) algorithm runs in $O(Kn^{2}mT^{3})$ time, where $K+1$ is the number of $\lambda$ values in the set $BP$ at the end of the algorithm. \end{theorem}
\begin{proof} The building, as well as the updating of the time-dependent residual network requires an $O(mT)$ running time, since for each of the time values $0\leq \theta \leq T$ all the $m$ arcs must be examined. Labels Setting (LS)  procedure investigates at most $nT$ adjacent node-time pairs for each of the node-time pairs which are removed from the list $L$ (i.e. $O(nT)$ times). Thus, the complexity of labels setting (LS) procedure is $O(n^{2}T^{2})$. Considering that in each of the iterations of the QDP procedure, for one of the time values, one arc is eliminated from the dynamic residual network, the algorithm end in at most $O(mT)$ iterations. On each of the iterations the procedure finds a quickest dynamic path with the complexity $O(n^{2}T^{2})$ and updates the time-dependent residual network in $O(mT)$ time. Thus, the total complexity of the quickest dynamic paths (QDP) procedure is $O(n^{2}mT^{3})$.\\
For each of the $K$ subintervals $[\lambda_{k},\lambda_{k+1}]$, $k=0,1,\cdots,K-1$ in which the interval $[0,\Lambda]$  of the parameter $\lambda$ values is partitioned in, the algorithm makes a call to procedure QDP. Consequently, the complexity of the parametric minimum dynamic flow (PmDF) algorithm is $O(Kn^{2}mT^{3})$. 
\end{proof}

\section{Example}
\label{sec:four}
In the discrete-time dynamic network presented in Figure~\ref{fig1}, node $1$ is the source node $s$ and node $4$ is the sink node $t$; the time horizon being set to $T=3$, i.e. $H=\{0, 1, 2, 3\}$. 

\begin{figure}[htbp]
  \begin{center}
    \subfigure[The discrete-time dynamic network $\bar{G}$. \label{fig1}]
		{\includegraphics[width=0.45\textwidth]{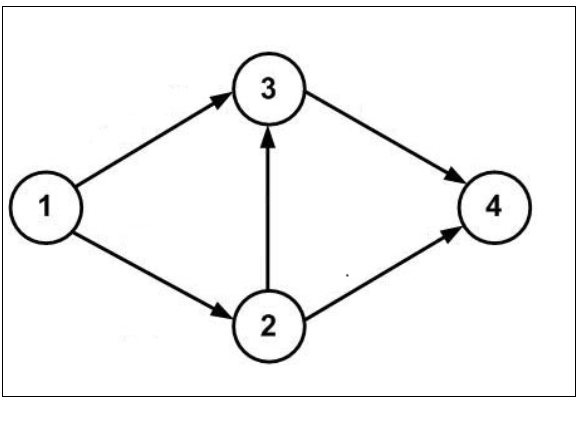}}
    \hfil
    \subfigure[The piecewise linear minimum flow over time value function for $\bar{G}$.\label{fig2}]
		{\includegraphics[width=0.4\textwidth]{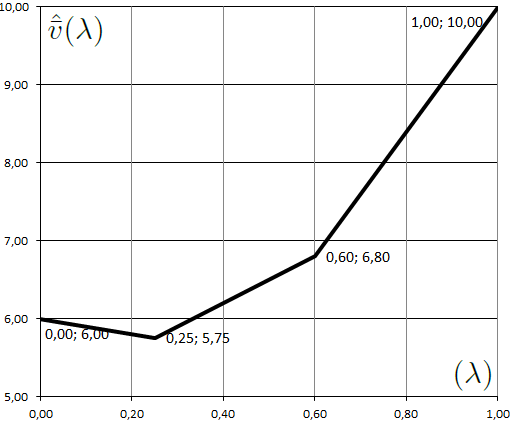}}
    \caption{Illustrating PmDF algorithm.}
    \label{fig}
  \end{center}
\end{figure}

For the interval of the parameter $\lambda$ values, set to $[0,1]$, i.e. $\Lambda=1$, the transit times $h(i,j;\theta)$, parametric lower bound functions $\bar{\ell}(i,j;\theta;\lambda)=\ell_{0}(i,j;\theta)+\lambda \cdot \mathcal{L}(i,j;\theta)$, upper bounds (capacities) $u(i,j;\theta)$ and feasible flow $\bar{f}(i,j;\theta)$ for all arcs in $\bar{G}$ are indicated in Table \ref{tab1}. 

\begin{table*}[h] 
\caption{\small Characteristics of dynamic network $\bar{G}$ presented in Figure \ref{fig1}}
\label{tab1}
\begin{center} \footnotesize
\begin{tabular}{l|lllcl} 
\hline
$(i,j)$ & $h(i,j;\theta)$ & $\ell_{0}(i,j;\theta)$ & $\mathcal{L}(i,j;\theta)$ & $u(i,j;\theta)$ & $\bar{f}(i,j;\theta)$ \\ \hline
$(1,2)$ & \begin{tabular}{ll} 1, & $\theta=0$ \\ 2, & $\theta \geq 1$ \end{tabular} & \begin{tabular}{ll} 3, & $\theta=0$ \\ 0, & $\theta \geq 1$ \end{tabular} & \begin{tabular}{ll} -2, & $\theta=0$ \\ 0, & $\theta \geq 1$ \end{tabular} & 5 & \begin{tabular}{ll} 5, & $\theta=0$ \\ 0, & $\theta \geq 1$ \end{tabular} \\ \hline
$(1,3)$ & \begin{tabular}{ll} 1, & $0\leq \theta <2$ \\ 2, & $\theta \geq 2$ \end{tabular} & \begin{tabular}{ll} 1, & $0\leq \theta <2$ \\ 0, & $\theta \geq 2$ \end{tabular} & \begin{tabular}{lll} 4, & $\theta=0$ \\ 1, & $ \theta=1$ \\ 0, & $\theta > 1$ \end{tabular} & 5 & \begin{tabular}{lll} 5, & $\theta =0$ \\ 2, & $\theta =1$ \\ 0, & $\theta \geq 1$ \end{tabular} \\ \hline
$(2,3)$ & \begin{tabular}{ll} 1,& $\theta \geq 0$ \end{tabular} & \begin{tabular}{ll} 0,& $\theta \geq 0$ \end{tabular} & \begin{tabular}{ll} 3, & $\theta=1$ \\ 0, & $\theta \neq 1$ \end{tabular} & 5 & \begin{tabular}{ll} 3, & $\theta=1$ \\ 0, & $\theta \neq 1$ \end{tabular} \\ \hline
$(2,4)$ & \begin{tabular}{ll} 1, & $0\leq \theta <2$ \\ 2, & $\theta \geq 2$ \end{tabular} & \begin{tabular}{ll} 0,& $\theta \geq 0$ \end{tabular} & \begin{tabular}{ll} 0,& $\theta \geq 0$ \end{tabular} & 5 & \begin{tabular}{ll} 2, & $\theta=1$ \\ 0, & $\theta \neq 1$ \end{tabular} \\ \hline
$(3,4)$ & \begin{tabular}{ll} 2, & $0\leq \theta <2$ \\ 1, & $\theta \geq 2$ \end{tabular} & \begin{tabular}{lll} 0, & $\theta=0$ \\ 2, & $ 0 < \theta \leq 2$ \\ 0, & $\theta > 2$ \end{tabular} & \begin{tabular}{ll} -2, & $ \theta =2$ \\ 0, & $\theta \neq 2$ \end{tabular} & 5 & \begin{tabular}{lll} 0, & $\theta =0$ \\ 5, & $0 < \theta \leq 2$ \\ 0, & $\theta > 2$ \end{tabular} \\ \hline
\end{tabular}	\end{center} \end{table*} 

\small
After the initialisation step, for $k=0$ and for the corresponding initial value of the parameter $\lambda_{0}=0$, procedure QDP is called for the first time. The successor vector is initialised for all nodes at all time values to $\sigma(i,\theta):=(0,0)$ and the Labels setting (LS) procedure is called for finding a quickest dynamic path in the time-dependent residual network $\bar{G}(\bar{f})$. 
The distance labels are initialised to $\tau(4,\theta):=0$, $\forall \theta \in \{ 0,1,2,3\}$ and the set of candidate nodes is set to $L:=\{(4,0),(4,1),(4,2),(4,3)\}$. After setting transit time labels and finding successor values for all node-time pairs, the procedure computes the minimum label of the source node,  $\bar{\tau}:=\min \{ \tau(1,0), \tau(1,1), \tau(1,2), \tau(1,3) \}=\min \{2, 1, \infty, \infty \}=1$ with $\bar{\theta}:=1$. Since $\bar{\tau}(4) \neq \infty$ the procedure LS ends with the variable $C$ keeping its initial value $C=1$. Based on successor vector, the quickest dynamic path $\bar{P}:=((1,1),(3,2),(2,1),(4,2))$ is built and its residual capacity $\bar{r}(\bar{P};\theta ; \lambda)=\alpha +(\lambda - \lambda_{k}) \cdot \beta$ is computed with $\alpha=\min \{ \alpha_{0}(1,3;1), \alpha_{0}(3,2;2), \alpha_{0}(2,4;1) \}=\min \{1,2,2 \}=1$ and $\beta=-1$. After the time-dependent residual network is updated and the parametric dynamic flow is decreased along the shortest dynamic path, the successor vector is reinitialised and procedure LS is called again. The next shortest dynamic path found by QDP procedure is $\bar{P}:=((1,0),(2,1),(4,2))$ with $\bar{r}(\bar{P};\theta ; \lambda)=1+\lambda$ and both the time-dependent residual network and the parametric dynamic flow are accordingly updated. Then QDP procedure finds the new dynamic path $\bar{P}:=((1,0),(3,1),(4,3))$ with $\alpha=3$ and $\beta=0$. Since $\beta_{0}(1,3;0)= -4 < \beta=0$, the upper limit $\lambda_{k+1}$ of the subinterval of the parameter is updated to $\lambda_{1}=\min \{ 1, (4-3)/(0+4) \}= 1/4$.  Finaly, after updating the time-dependent residual network and the parametric dynamic flow, based on successor vector found by LS procedure, QDP finds the last path $\bar{P}:=((1,0),(2,1),(3,2),(4,3))$ with $\alpha=1$ and $\beta=1$. The validity of the upper limit of the subinterval of the parameter values is tested and it is maintained unchanged since for the arc $(2,3)$ at $\theta=1$, $\beta_{0}(2,3;1)=-4<\beta=1$ but $\lambda_{1}=\min\{1/4,(4-1)/(1+4)\}=\min\{1/4, 3/5\}=1/4$. At this point, after the updating step, no other dynamic path can be found in the time-dependent residual network so that $C:=0$ is set, the breakpoints list is updated to $BP=\{0, 1/4\}$ and the first iteration ends with the minimum parametric flow over time $\hat{\bar{f}}_{0}(i,j;\theta;\lambda)$ computed for the subinterval $[0, 1/4]$ of the parameter $\lambda$ values. 
The PmDF algorithm increments variable $k$ to the value $k=1$ and a next iteration is performed. The evolution of the algorithm is presented in Table \ref{tab2}. 
The piecewise linear minimum flow over time value function for the discrete dynamic network $\bar{G}$, computed by \textit{Parametric minimum Dynamic Flow} (PmDF) algorithm, is presented in Figure~\ref{fig2}.

\begin{table*}[h]
\caption{\small The evolution of PmDF algorithm for the dynamic network $\bar{G}$}
\label{tab2}
\begin{center} \footnotesize
\begin{tabular}{c c c c c c}
\hline
 & $k$& $\lambda_{k}$& $\bar{P}$& $\bar{r}_{k}(\bar{P};\theta;\lambda)$& $\lambda_{k+1}$\\\hline
 & &&$(1,1),(3,2),(2,1),(4,2)$&$1-\lambda$&$1$\\\cline{4-6}
$Iteration \; 1:$ &$0$&$0$&$(1,0),(2,1),(4,2)$&$1+\lambda$&$1$\\\cline{4-6}
 & & &$(1,0),(3,1),(4,3)$&$3$&$1/4$\\\cline{4-6}
 & & & $(1,0),(2,1),(3,2),(4,3)$&$1+\lambda$&$1/4$ \\\hline
  
 & &&$(1,1),(3,2),(2,1),(4,2)$&$1-\lambda$&$1$\\\cline{4-6}
$Iteration \; 2:$ &$1$&$1/4$&$(1,0),(2,1),(4,2)$&$1+\lambda$&$1$\\\cline{4-6}
 & & &$(1,0),(3,1),(4,3)$&$4-4 \lambda$&$1$\\\cline{4-6}
 & & & $(1,0),(2,1),(3,2),(4,3)$&$1+\lambda$&$3/5$ \\\hline 

 & &&$(1,1),(3,2),(2,1),(4,2)$&$1-\lambda$&$1$\\\cline{4-6}
$Iteration \; 3:$ &$2$&$3/5$&$(1,0),(2,1),(4,2)$&$1+\lambda$&$1$\\\cline{4-6}
 & & &$(1,0),(3,1),(4,3)$&$4-4 \lambda$&$1$\\\cline{4-6}
 & & & $(1,0),(2,1),(3,2),(4,3)$&$4-4 \lambda$&$1$ \\\hline 
\end{tabular}
\end{center}
\end{table*}

\bibliographystyle{abbrvnat}
\bibliography{references}
\label{sec:biblio}

\end{document}